\documentclass[11pt]{article}
\usepackage{amsmath,amssymb,bm,amsthm,array,tikz,enumerate}
\usepackage{mathtools}
\usepackage{url}
\usepackage{lscape}
\usepackage[margin=30truemm]{geometry}
 
\usepackage{float}
\usepackage{makecell}
\usepackage{blkarray}
\usepackage{multirow}
\usepackage[hidelinks]{hyperref}

\usetikzlibrary{positioning,graphs}

\numberwithin{equation}{section}

\DeclareMathOperator{\Spec}{Spec}

\DeclareMathOperator{\Aut}{Aut}

\newcommand{\MC}[1]{\mathcal{#1}}
\newcommand{\MB}[1]{\mathbb{#1}}

\newcommand{\BM}[1]{{\bm #1}}

\newtheorem{thm}{Theorem}[section]
\newtheorem{lem}[thm]{Lemma}
\newtheorem{pro}[thm]{Proposition}
\newtheorem{cor}[thm]{Corollary}

\theoremstyle{definition}

\allowdisplaybreaks

\newcommand{\G}{\Gamma}
\newcommand{\D}{\Delta}

\newcommand{\init}{\, \bm{\mathrm{j}}}

\title
{
Arc search in graphs via Szegedy walks \\
}

\author{
\hfill
\makebox[0.4\textwidth][c]{
Sho Kubota\thanks{
Department of Mathematics Education,
Aichi University of Education,
1 Hirosawa, Igaya-cho, Kariya, Aichi 448-8542, Japan.
\texttt{skubota@auecc.aichi-edu.ac.jp}}
}
\hfill
\and
\hfill
\makebox[0.4\textwidth][c]{
Kiyoto Yoshino\thanks{
Department of Information Science, 
Faculty of Science, 
Toho University,
2-2-1 Miyama, Funabashi, Chiba 274-8510, Japan
\texttt{kiyoto.yoshino@is.sci.toho-u.ac.jp}
}
}
\hfill
}
\date{}

\begin{document}
\maketitle

\begin{abstract}
This paper studies the search for a single arc in a graph using the Szegedy walk.
Arc search can be interpreted as finding a quantum particle not only in its position but also with a specific internal state.
The quantum walk employed in this study is essentially the model proposed by Segawa and Yoshie for the purpose of edge search.
First, we investigate how the symmetry of a graph is reflected in its time evolution matrix, and provide a sufficient condition under which the success probability of the search is independent of the marked arc.
In particular, we prove that if a graph is arc-transitive, the success probability is independent of the choice of the marked arc.
Next, we analyze path and cycle graphs and show that the quantum search is ineffective for these graphs, whereas it performs well for complete bipartite graphs $K_{n,n}$.
These results provide a theoretical foundation for studying arc and edge searches on various graphs, while also suggesting new problems concerning the eigenvalue analysis of edge-signed graphs in spectral graph theory.
\vspace{8pt} 
\end{abstract}

\section{Introduction}

Quantum search has been widely studied as one of the representative applications of quantum computing and is known to provide dramatic speedups over classical search algorithms.
In particular, Grover's algorithm~\cite{grover1996fast} demonstrated that a marked item can be found in a database in $O(\sqrt{N})$ steps (more precisely, $\Theta(\sqrt{N})$ steps), which has attracted significant attention.
Subsequently, it has been shown that Grover's algorithm can be generalized to efficiently search for multiple marked items~\cite{boyer1998tight}.
These developments have established quantum search as a versatile tool for a variety of algorithms.
This is discussed, for example, in the textbook by Nielsen and Chuang~\cite{nielsen2010quantum}.

In parallel with this trend,
search algorithms based on quantum walks began to be studied,
such as~\cite{shenvi2003quantum, szegedy2004quantum}.
Furthermore, subsequent studies have established new frameworks for acceleration that exploit graph structures,
including applications to element distinctness~\cite{ambainis2007quantum} and triangle detection~\cite{magniez2007quantum}.

This paper formulates the search problem from a perspective different from conventional quantum search.
The framework for defining quantum walks on finite-dimensional vector spaces indexed by symmetric arc sets has been extensively studied, but most existing search algorithms mark vertices in a graph as the search targets~\cite{chakraborty2016spatial, krovi2016quantum}.
Search algorithms for finding edges in graphs have also been discussed recently in~\cite{segawa2021quantum, yoshie2022quantum}.
In contrast, this study considers the setting in which an arc is marked as the search target.
In this framework, the search can be interpreted as not only identifying the vertex where the walker is found but also simultaneously detecting internal degrees of freedom, such as the direction of motion or the spin state.

First, we study the relationship between search behavior in discovering a marked arc and symmetry of graphs.
We refer to the probability of finding a marked item, whether a vertex or an arc, as the \emph{success probability}.
In the case of vertex search on vertex-transitive graphs, it is widely understood that the probability of finding a marked vertex is independent of which vertex is marked.
Similarly, in arc search on arc-transitive graphs, the success probability is likewise expected to be independent of the choice of marked arc.
However, this fact is not necessarily obvious and requires a rigorous mathematical justification.
One of our main results is to clarify conditions for symmetry of graphs that guarantee equal probabilities of discovering marked arcs.
In particular, we prove that if a graph is arc-transitive, the success probability is independent of the choice of the marked arc.
See Corollary~\ref{0820-1} for details.

Next, we analyze the search behavior and computational complexity on path graphs, cycle graphs, and complete bipartite graphs.
We show that our search algorithm is not effective for path or cycle graphs, whereas it performs well on complete bipartite graphs $K_{n,n}$.
More specifically, the success probabilities on the path graphs $P_n$ and the cycle graphs $C_n$ are $\frac{1}{2n-2}$ and $\frac{1}{2n}$, respectively, and are independent of the measurement time.

For the complete bipartite graphs $K_{n,n}$, let $p^*$ and $q^*$ denote the probabilities of finding the marked arc and its inverse arc, respectively, at an appropriately chosen measurement time.
We show that, for every $n\geq2$, we have
\[ \left|p^*-\frac12\right|\leq\frac{3+2\sqrt2}{n},\qquad
\left|q^*-\frac12\right|\leq\frac{3+2\sqrt2}{n}. \]
In particular, both probabilities converge to $1/2$.
Furthermore, the probability $p^*+q^*$ of finding the edge containing the marked arc converges to $1$.
More precisely, for every $n\geq2$, we have
\[ p^*+q^*\geq1-\frac{6+4\sqrt2}{n}. \]
The corresponding measurement time is $\Theta(n)$.
Since $K_{n,n}$ has $2n^2$ arcs, this gives a quadratic speedup for arc search.
\section{Preliminaries}

See \cite{godsil2013algebraic} for basic terminology related to graphs.
Let $\G =(V, E)$ be a graph with the vertex set $V$ and the edge set $E$.
Throughout this paper, we assume that graphs are simple and finite,
i.e., $|V| < \infty$ and $E \subset \{\{x,y\} \subset V \mid x \neq y\}$.
We also assume that graphs have no isolated vertices.
Define $\MC{A} = \MC{A}(\G)=\{ (x, y), (y, x) \mid \{x, y\} \in E \}$,
which is the set of the \emph{symmetric arcs} of $\G$.
The origin $x$ and terminus $y$ of $a=(x, y) \in \MC{A}$ are denoted by $o(a)$ and $t(a)$, respectively.
We write the inverse arc of $a$ as $a^{-1}$.

Before defining the time evolution matrix used in the search,
we introduce several matrices related to Grover walks.
Note that Grover walks are also referred to as arc-reversal walks or arc-reversal Grover walks,
and they are known as a special case of bipartite walks \cite{Chen2024Hamiltonians}.
Let $\G = (V, E)$ be a graph, and set $\MC{A} = \MC{A}(\Gamma)$.
The \emph{boundary matrix} $d = d(\G) \in \MB{C}^{V \times \MC{A}}$ is defined by
$d_{x,a} = \frac{1}{\sqrt{\deg x}} \delta_{x, t(a)}$,
where $\delta_{a,b}$ is the Kronecker delta.
By directly calculating the entries based on the definition of matrix multiplication, it follows that
\begin{equation*} %
dd^* = I,
\end{equation*}
where $I$ is the identity matrix.
The \emph{shift matrix} $S = S(\G) \in \MB{C}^{\MC{A} \times \MC{A}}$
is defined by $S_{a, b} = \delta_{a,b^{-1}}$.
Clearly, $S^2 = I$.
The time evolution matrix $U = U(\G) \in \MB{C}^{\MC{A} \times \MC{A}}$ for the Grover walk is defined by
\[ U := S(2d^*d-I), \]
but we do not use $U$ itself.
Instead, we use a modified form of $U$,
which is introduced in the next section, to perform the search.

Furthermore, we describe relationships between symmetry of graphs and Grover walks.
For the convenience of a broad readership,
we begin by introducing automorphisms of graphs and their associated permutation matrices.
Let $\G = (V, E)$ be a graph.
A mapping $g: V \to V$ is an \emph{automorphism} of $\G$
if $g$ is bijective, and $\{x,y\} \in E$
if and only if $\{g(x), g(y) \} \in E$.
We denote the set of all automorphisms of $\G$ by $\Aut(\G)$.

Let $g$ be an automorphism of a graph $\G$.
Since $g$ is bijective,
it defines the permutation matrix $M_{g} \in \MB{C}^{V \times V}$ given by $(M_{g})_{x,y} = \delta_{x, g(y)}$.
Basic properties of permutation matrices are
\begin{equation} \label{1121-1}
M_{g}^\top = M_{g}^{-1} = M_{g^{-1}},
\end{equation}
and
\begin{equation*} %
M_{g}e_x = e_{g(x)}
\end{equation*}
for $x \in V$.

The action of automorphisms of a graph is naturally extended to the set of symmetric arcs.
Let $g$ be an automorphism of a graph $\G$,
and let $\MC{A} = \MC{A}(\G)$ be the symmetric arc set.
Define $\tilde{g}: \MC{A} \to \MC{A}$ by
$\tilde{g}(\left(x,y)\right) = (g(x), g(y))$.
Clearly, $\tilde{g}$ is a bijection.
Thus, the permutation matrix $N_{\tilde{g}} \in \MB{C}^{\MC{A} \times \MC{A}}$ is defined by $(N_{\tilde{g}})_{a,b} = \delta_{a, \tilde{g}(b)}$,
and it holds that 
\begin{equation*} %
N_{\tilde{g}}^{\top} = N_{\tilde{g}}^{-1} = N_{\tilde{g}^{-1}}
\end{equation*}
as in Equality~(\ref{1121-1}).
Moreover, we have
\begin{equation*} %
N_{\tilde{g}}e_a = e_{\tilde{g}(a)}
\end{equation*}
for $a \in \MC{A}$.
The following statements describe relationships between permutation matrices induced by automorphisms and matrices associated with Grover walks.
See~\cite[Lemmas~3.1 and~3.2]{kubota2025circulant} for the proofs.

\begin{lem} \label{0814-1}
Let $g$ be an automorphism of a graph $\G$.
Then we have
\begin{align}
dN_{\tilde{g}} &= M_g d, \label{0817-1} \\
SN_{\tilde{g}} &= N_{\tilde{g}} S. \label{0817-2}
\end{align}
\end{lem}

\section{Quantum search and symmetry of graphs}

In this section, we begin by presenting the quantum search model introduced by Segawa and Yoshie~\cite{segawa2021quantum}.
Let $\G = (V, E)$ be a graph, and let $\MC{A} = \MC{A}(\G)$ denote the set of symmetric arcs.
{We equip the vector space $\mathbb{C}^A$ with the standard inner product and regard the resulting Hilbert space as the total state space of the quantum walk.}
A mapping $\sigma: \MC{A} \to \{\pm 1\}$ is said to be a {\it sign function} if $\sigma(a) = -1$ implies $\sigma(a^{-1}) = 1$ for all $a \in \MC{A}$.
In other words,
it is not allowed for both $a$ and $a^{-1}$ to be assigned negative signs.
An arc $a$ with $\sigma(a) = -1$ is called a {\it marked arc}.
Our objective is to perform a quantum search for marked arcs.
To this end, we use a time evolution matrix that is a slight modification of the boundary matrix introduced in the previous section.
For a sign function $\sigma$, we define the boundary matrix $d_{\sigma} \in \MB{C}^{V \times \MC{A}}$ associated with $\sigma$ by
\[ (d_{\sigma})_{x,a} := \frac{\sigma(a)}{\sqrt{\deg x}}\delta_{x,t(a)}. \]
In addition, we define the time evolution matrix $U_{\sigma} \in \MB{C}^{\MC{A} \times \MC{A}}$ associated with $\sigma$ by
\[ U_{\sigma} := S(2d_{\sigma}^* d_{\sigma} - I). \]
This quantum walk is a special case of the Szegedy walk, and therefore the unitarity of $U_{\sigma}$ is guaranteed.
For the proof, see~\cite[Proposition~1]{segawa2015spectral}.
Yoshie and Yoshino calculate the $(a,b)$-entry of $U_{\sigma}$ for $a,b \in \MC{A}$ as follows~\cite{yoshie2022quantum}:
\begin{equation} \label{0819-1}
(U_{\sigma})_{a,b} = \frac{2\sigma(a^{-1})\sigma(b)}{\deg t(b)} \delta_{o(a), t(b)} - \delta_{a,b^{-1}}.
\end{equation}

Using this time evolution matrix, we describe the quantum search for marked arcs.
Throughout this paper, we take the initial state to be the normalized all-ones vector
\[
\init := \frac{1}{\sqrt{|\MC{A}|}} [1,1,\dots,1]^{\top} \in \MB{C}^{\MC{A}}.
\]
Starting from $\init$, the state after $\tau$ steps is $U_{\sigma}^{\tau}\init$.
{For an arc $b\in A$, let $\Pi_b$ denote the orthogonal projection onto $\operatorname{span}\{\BM{e}_b\}$.}
If we measure $U_{\sigma}^{\tau}\init$ in the standard basis $\{e_b \mid b \in \MC{A}\}$, then the probability of observing an arc $b \in \MC{A}$ is
\[
{\| \Pi_b U_{\sigma}^{\tau}\init \|^2 =} |(U_{\sigma}^{\tau}\init)_b|^2.
\]
Accordingly, the success probability at time $\tau$ is
\[
\sum_{\substack{b \in \MC{A} \\ \sigma(b) = -1}} |(U_{\sigma}^{\tau}\init)_b|^2.
\]

In the remainder of this section, we study the relationship between graph symmetries and the success probability of the quantum search. 
We define the {\it oracle matrix} $\MC{O}_{\sigma} \in \MB{C}^{\MC{A} \times \MC{A}}$ by
\[ (\MC{O}_{\sigma})_{a,b} := \sigma(a)\delta_{a,b}. \]

\begin{lem} \label{0812-1}
Let $\sigma$ be a sign function.
Then we have $d_{\sigma} = d \MC{O}_{\sigma}$.
\end{lem}
\begin{proof}
Indeed, for a vertex $x$ and an arc $a$, we have
\[ (d \MC{O}_{\sigma})_{x,a}
= \sum_{z \in \MC{A}} d_{x,z} (\MC{O}_{\sigma})_{z,a}
= \sum_{z \in \MC{A}} \frac{1}{\sqrt{\deg x}} \delta_{x, t(z)} \cdot \sigma(z) \delta_{z,a}
= \frac{\sigma(a)}{\sqrt{\deg x}} \delta_{x, t(a)}
= (d_{\sigma})_{x,a},
\]
as claimed.
\end{proof}

Note that the time-evolution matrix $U_{\sigma}$ for quantum search can also be expressed using the coin matrix $C := 2d^* d - I$ of the Grover walk as
\[
U_{\sigma} = S\MC{O}_{\sigma} C \MC{O}_{\sigma}.
\]
Indeed, Lemma~\ref{0812-1} together with $\MC{O}_{\sigma}^* = \MC{O}_{\sigma}$ and $\MC{O}_{\sigma}^2 = I$ implies that
\begin{align*}
U_{\sigma} &= S(2d_{\sigma}^* d_{\sigma} - I)
= S(2(d\MC{O}_{\sigma})^* d\MC{O}_{\sigma} - \MC{O}_{\sigma}^2) \\
&= S(2 \MC{O}_{\sigma}d^* d\MC{O}_{\sigma} - \MC{O}_{\sigma}^2)
= S\MC{O}_{\sigma}(2d^* d - I)\MC{O}_{\sigma}
= S\MC{O}_{\sigma} C \MC{O}_{\sigma}.
\end{align*}
In this paper,
we primarily focus on sign functions that mark only a single arc.
For an arc $z \in \MC{A}$, we define the specific sign function $\sigma_z : \MC{A} \to \{\pm 1\}$ by
\[ \sigma_z(w) := \begin{cases}
1 \quad &\text{if $w \neq z$}, \\
-1 \quad &\text{if $w = z$}.
\end{cases} \]
In this case, the success probability at time $\tau$ is
\[
|(U_{\sigma_z}^{\tau}\init)_z|^2.
\]

\begin{lem} \label{0812-2}
Let $\G$ be a graph, and let $a$ and $b$ be arcs of $\G$.
If an automorphism $g$ of $\G$ satisfies $\tilde{g}(a) = b$,
then we have
\[ N_{\tilde{g}}^* \MC{O}_{\sigma_b}N_{\tilde{g}} = \MC{O}_{\sigma_a}. \]
\end{lem}

\begin{proof}
For arcs $z$ and $w$, we have
\begin{align*}
(N_{\tilde{g}}^* \MC{O}_{\sigma_b}N_{\tilde{g}})_{z,w}
&= e_z^* (N_{\tilde{g}}^* \MC{O}_{\sigma_b}N_{\tilde{g}}) e_{w}
= (N_{\tilde{g}} e_{z} )^* \MC{O}_{\sigma_b} (N_{\tilde{g}} e_{w}) \\
&= e_{\tilde{g}(z)}^* \MC{O}_{\sigma_b} e_{\tilde{g}(w)} 
= (\MC{O}_{\sigma_b})_{\tilde{g}(z), \tilde{g}(w)}
= \sigma_b(\tilde{g}(z)) \delta_{z,w}.
\end{align*}
In addition, 
\begin{align*}
\sigma_b(\tilde{g}(z)) = -1
&\iff \tilde{g}(z) = b = \tilde{g}(a) \\
&\iff z = a \\
&\iff \sigma_a(z) = -1,
\end{align*}
which implies
\[ \sigma_b(\tilde{g}(z)) \delta_{z,w}
= \sigma_a(z) \delta_{z,w} = (\MC{O}_{\sigma_a})_{z,w}, \]
as claimed.
\end{proof}

\begin{thm} \label{0813-1}
Let $\G$ be a graph, and let $a$ and $b$ be arcs of $\G$.
If an automorphism $g$ of $\G$ satisfies $\tilde{g}(a) = b$,
then we have
\[ N_{\tilde{g}}^* U_{\sigma_b}N_{\tilde{g}} = U_{\sigma_a}. \]
\end{thm}

\begin{proof}
We will show that $U_{\sigma_b}N_{\tilde{g}} = N_{\tilde{g}} U_{\sigma_a}$.
We have
\begin{align*}
U_{\sigma_b}N_{\tilde{g}}
&= S(2d^*_{\sigma_b}d_{\sigma_b} - I) N_{\tilde{g}} \\
&= 2S \MC{O}_{\sigma_b} d^* d \MC{O}_{\sigma_b} N_{\tilde{g}} - S N_{\tilde{g}} \tag{by Lemma~\ref{0812-1}} \\
&= 2S \MC{O}_{\sigma_b} d^* d N_{\tilde{g}} \MC{O}_{\sigma_a} -  N_{\tilde{g}}S \tag{by Lemma~\ref{0812-2} and~\eqref{0817-2}} \\
&= 2S \MC{O}_{\sigma_b} d^* M_{g} d \MC{O}_{\sigma_a} -  N_{\tilde{g}}S \tag{by~\eqref{0817-1}} \\
&= 2S \MC{O}_{\sigma_b} N_{\tilde{g}} d^*  d \MC{O}_{\sigma_a} -  N_{\tilde{g}}S \tag{by~\eqref{0817-1}} \\
&= 2S N_{\tilde{g}} \MC{O}_{\sigma_a} d^*  d \MC{O}_{\sigma_a} -  N_{\tilde{g}}S \tag{by Lemma~\ref{0812-2}} \\
&= 2 N_{\tilde{g}} S \MC{O}_{\sigma_a} d^*  d \MC{O}_{\sigma_a} -  N_{\tilde{g}}S \tag{by~\eqref{0817-2}} \\
&= N_{\tilde{g}}S (2 d^*_{\sigma_a} d_{\sigma_a} - I) \tag{by Lemma~\ref{0812-1}} \\
&= N_{\tilde{g}} U_{\sigma_a},
\end{align*}
as claimed.
\end{proof}

By applying Theorem~\ref{0813-1},
we can show that if there exists an automorphism that maps an arc $a$ to an arc $b$, then the probabilities of finding $a$ and $b$ are equal.

\begin{cor} \label{1007-1}
Let $\G$ be a graph, and let $a$ and $b$ be arcs of $\G$.
If an automorphism $g$ of $\G$ satisfies $\tilde{g}(a) = b$,
then for any positive integer $\tau$, we have
\[ (U_{\sigma_a}^{\tau} \init)_a = (U_{\sigma_b}^{\tau} \init)_b. \]
In particular, 
\[ |(U_{\sigma_a}^{\tau} \init)_a|^2 = |(U_{\sigma_b}^{\tau} \init)_b|^2. \]
\end{cor}

\begin{proof}
By Theorem~\ref{0813-1}, we have

\[
(U_{\sigma_a}^{\tau} \, \bm{\mathrm{j}})_a
= \langle e_a,  U_{\sigma_a}^{\tau} \, \bm{\mathrm{j}} \rangle
= \langle e_a, N_{\tilde{g}}^* U_{\sigma_b}^{\tau} N_{\tilde{g}} \, \bm{\mathrm{j}} \rangle
= \langle N_{\tilde{g}} e_a,  U_{\sigma_b}^{\tau} N_{\tilde{g}} \, \bm{\mathrm{j}} \rangle.
\]
Since $N_{\tilde{g}}$ is a permutation matrix,
it follows that $N_{\tilde{g}} \, \init = \init$.
Hence,

\[ (U_{\sigma_a}^{\tau} \, \bm{\mathrm{j}})_a
= \langle e_{\tilde{g}(a)}, U_{\sigma_b}^{\tau} \, \bm{\mathrm{j}} \rangle
= \langle e_{b}, U_{\sigma_b}^{\tau} \, \bm{\mathrm{j}} \rangle
= (U_{\sigma_b}^{\tau} \init)_b,
\]
as claimed.
\end{proof}

For a graph $\G$,
the automorphism group $\Aut(\G)$ acts on the symmetric arc set $\MC{A}$ via the rule $(g, a) \mapsto \tilde{g}(a)$,
where $g \in \Aut(\G)$ and $a \in \MC{A}$.
If this action has only one orbit,
the graph $\G$ is said to be {\it arc-transitive}.
{
Note that a concept similar to arc-transitivity is edge-transitivity, which means that the action $\Aut(\G) \times E \to E$ defined by $(g,\{x,y\}) \mapsto \{g(x),g(y)\}$ is transitive.
Roughly speaking, the difference between arc-transitivity and edge-transitivity is whether ordered pairs or unordered pairs can be mapped to each other.
The complete bipartite graph $K_{m,n}$ is always edge-transitive, but it is not arc-transitive when $m\neq n$.
}

\begin{cor} \label{0820-1}
Let $\G$ be a graph,
and let $\MC{A}_1, \dots, \MC{A}_t$ be the orbits of $\MC{A}$ under the action of $\Aut(\G)$.
If $a, b \in \MC{A}_i$ for some $i$,
then
\[ |(U_{\sigma_a}^{\tau} \init)_a|^2 = |(U_{\sigma_b}^{\tau} \init)_b|^2. \]
In particular, if $\G$ is arc-transitive,
then for any arcs $a$ and $b$, we have
\[ |(U_{\sigma_a}^{\tau} \init)_a|^2 = |(U_{\sigma_b}^{\tau} \init)_b|^2. \]
\end{cor}

Of course, the probabilities of discovering arcs from different orbits are generally not equal.

{
Here, we provide some remarks on the nontriviality of Theorem~\ref{0813-1} and its corollaries.
These results are based on the fact that the matrices associated with the Grover walk satisfy intertwining relations with the permutation matrices induced by automorphisms of graphs.
Indeed, we have $d N_{\tilde{g}}=M_g d$ as stated in Lemma~\ref{0814-1}.
This relation implies that $N_{\tilde{g}}C = C N_{\tilde{g}}$ for the Grover coin $C := 2d^*d-I$.
This commutativity plays an important role in establishing Theorem~\ref{0813-1}.
On the other hand, as we saw in Lemma~\ref{0812-2}, when a marked arc is mapped by an automorphism, the corresponding oracle matrix is also mapped accordingly.
It is by combining such commutativity with the transformation rule for the oracle matrix that we obtain Theorem~\ref{0813-1}, namely, the relation between the time evolution matrices corresponding to different marked arcs.
Such a relation does not automatically hold for a general quantum coin.
For example, consider a generalized boundary matrix $d'$ having different weights on the arcs, and define a different coin matrix $C' :=2(d')^* d' - I$ using $d'$.
If these weights are not preserved by an automorphism, then $N_{\tilde g} C' = C' N_{\tilde g}$ does not hold in general.
Therefore, the relation stated in Theorem~\ref{0813-1} does not automatically hold for general Szegedy walks, but is obtained by combining the high symmetry of the Grover walk with the transformation rule specific to the oracle matrix used in this study.
Furthermore, the fact that the initial state $\init$ of the quantum search satisfies $P \init = \init$ for any permutation matrix $P$ is another reason why Corollaries~\ref{1007-1} and~\ref{0820-1} hold.
}

\section{Arc search on cycle graphs and path graphs}

In this section, we show that the quantum search model defined in the previous section does not perform well on path graphs and cycle graphs.
An intuitive reason for this is that each vertex in these graphs has degree at most $2$, so no interference of probability amplitudes arises from superposition, and the probability of finding a marked arc remains constant over time.

\begin{pro}
Let $n \geq 3$, and let $a$ be an arc of $\MC{A} := \MC{A}(C_n)$,
and let $U_{\sigma_a} := U_{\sigma_a}(C_n)$.
Then for any non-negative integer $\tau$ and any arc $w$,
we have
\[ (U_{\sigma_a}^{\tau} \init)_w = \pm \frac{1}{\sqrt{2n}}. \]
In particular,
\[ |(U_{\sigma_a}^{\tau} \init)_a|^2 = \frac{1}{2n}. \]
\end{pro}

\begin{proof}
We prove the statement by induction on $\tau$.
The case $\tau = 0$ is obvious.
Assume that the statement holds for $\tau - 1$.
Then we have
\begin{align*}
(U_{\sigma_a}^{\tau} \init)_w
&= (U_{\sigma_a} U_{\sigma_a}^{\tau-1} \init)_w
= \sum_{z \in \MC{A}} (U_{\sigma_a})_{w,z} (U_{\sigma_a}^{\tau-1} \init)_{z} \\
&= \frac{1}{\sqrt{2n}} \sum_{z \in \MC{A}} \varepsilon(z) \left\{ \frac{2 \sigma_{a}(w^{-1})\sigma_{a}(z)}{2} \delta_{o(w), t(z)} - \delta_{w, z^{-1}} \right\}, \tag{by~\eqref{0819-1}}
\end{align*}
where $\varepsilon(z) \in \{\pm 1\}$.
Since $\deg o(w) = 2$,
there are exactly two arcs whose terminus is $o(w)$.
Let $w'$ denote the one that is not $w^{-1}$.
Then,
\begin{align*}
&\sum_{z \in \MC{A}} \varepsilon(z) \left\{ \frac{2 \sigma_{a}(w^{-1})\sigma_{a}(z)}{2} \delta_{o(w), t(z)} - \delta_{w, z^{-1}} \right\} \\
= \, \,  & \varepsilon(w^{-1}) \left( \sigma_{a}(w^{-1}) \sigma_{a}(w^{-1}) - 1 \right) + \varepsilon(w')\sigma_{a}(w^{-1}) \sigma_{a}(w') \\
= \, \,  & \varepsilon(w')\sigma_{a}(w^{-1}) \sigma_{a}(w'),
\end{align*}
and hence
\[ (U_{\sigma_a}^{\tau} \init)_w
= \frac{1}{\sqrt{2n}} \cdot \varepsilon(w') \sigma_{a}(w^{-1}) \sigma_{a}(w') = \pm \frac{1}{\sqrt{2n}}. \]
This completes the induction.
\end{proof}

A similar statement can be proved for path graphs, with minor modifications.

\begin{pro}
Let $n \geq 2$, and let $a$ be an arc of $\MC{A} := \MC{A}(P_n)$,
and let $U_{\sigma_a} := U_{\sigma_a}(P_n)$.
Then for any non-negative integer $\tau$ and any arc $w$,
we have
\[ (U_{\sigma_a}^{\tau} \init)_w = \pm \frac{1}{\sqrt{2n-2}}. \]
In particular,
\[ |(U_{\sigma_a}^{\tau} \init)_a|^2 = \frac{1}{2n-2}. \]
\end{pro}

\begin{proof}
We prove the statement by induction on $\tau$.
The case $\tau = 0$ is obvious.
Assume that the statement holds for $\tau - 1$.
Then we have
\begin{equation} \label{0818-1}
(U_{\sigma_a}^{\tau} \init)_w
= \frac{1}{\sqrt{2n-2}} \sum_{z \in \MC{A}} \varepsilon(z) \left\{ \frac{2 \sigma_{a}(w^{-1})\sigma_{a}(z)}{\deg o(w)} \delta_{o(w), t(z)} - \delta_{w, z^{-1}} \right\},
\end{equation}
where $\varepsilon(z) \in \{\pm 1\}$.
We consider path graphs, so $\deg o(w) \in \{1,2\}$.
If $\deg o(w) = 2$, the claim can be proved by the same calculation as in the case of cycle graphs.
If $\deg o(w) = 1$, the only arc whose terminus is $o(w)$ is $w^{-1}$.
Thus, the right-hand side of Equality~\eqref{0818-1} becomes
\[ \frac{1}{\sqrt{2n-2}} \varepsilon(w^{-1}) \left( \frac{2 \sigma_{a}(w^{-1})\sigma_{a}(w^{-1})}{1} - 1 \right)
= \frac{\varepsilon(w^{-1})}{\sqrt{2n-2}}, \]
which completes the induction.
\end{proof}

\section{Arc search on complete bipartite graphs} \label{0918-1}

In this section, we consider the problem of finding a single arc in the complete bipartite graph $K_{n,n}$ with $n \geq 2$.
The case $K_{1,1} \cong P_2$ is covered by the results of Section~4.
This study does not focus on the complete graphs, since their analysis has already been thoroughly carried out in the work of Segawa and Yoshie~\cite{segawa2021quantum} on edge search.
Therefore, the behavior of arc search on complete graphs can be regarded as essentially understood, because it is expected to show the same dynamics as the edge search analyzed in~\cite{segawa2021quantum}.
A complete bipartite graph is arc-transitive when the sizes of its partite sets are equal, so by Corollary~\ref{0820-1} we may, without loss of generality, assume that an arbitrary arc is marked.
In contrast to the previous section, we will see that quantum search is effective on complete bipartite graphs.

To evaluate the success probability, it is necessary to obtain the eigenvalues and eigenvectors of the time evolution matrix.
However, as we will see below, this task can be reduced to finding the eigenvalues and eigenvectors of edge-signed graphs.

Let $\G = (V, E)$ be a graph with a sign function $\sigma$.
We define the {\it discriminant matrix} $T_{\sigma} = T_{\sigma}(\G) \in \MB{C}^{V \times V}$ by
\[ T_{\sigma} := d_{\sigma} S d_{\sigma}^*. \]

\begin{thm}[{\cite[Proposition~1]{higuchi2014spectral}}]
Let $\G$ be a graph with a sign function $\sigma$,
and let $U_{\sigma} := U_{\sigma}(\G)$ and $T_{\sigma} := T_{\sigma}(\G)$.
Then it holds that
\[ \{ e^{\pm i \theta_{\lambda}} \mid \lambda \in \Spec(T_{\sigma}) \cap (-1,1) \} \subset \Spec(U_{\sigma}), \]
where $\theta_{\lambda} = \cos^{-1} \lambda$.
In addition, the normalized eigenvector of $U_{\sigma}$ associated with $e^{\pm i \theta_{\lambda}}$ is given by
\[ \varphi_{\pm \lambda}
= \frac{1}{\sqrt{2}|\sin \theta_{\lambda}|} (d_{\sigma}^* - e^{\pm i \theta_{\lambda}} S d_{\sigma}^*) \BM{f}, \]
where $\BM{f}$ is the normalized eigenvector of $T_{\sigma}$ associated with $\lambda$.
\end{thm}

In fact, the discriminant matrix $T_{\sigma}$ appearing in the above theorem can be interpreted as the normalized adjacency matrix of an edge-signed graph:

\begin{lem} \label{1016-1}
Let $\G = (V, E)$ be a graph with a sign function $\sigma$,
and let $T_{\sigma} := T_{\sigma}(\G)$.
Then we have
\[ (T_{\sigma})_{x,y} = \begin{cases}
\frac{\sigma(x,y)\sigma(y,x)}{\sqrt{\deg x} \sqrt{\deg y}} \quad &\text{if $\{x,y\} \in E$,} \\
0, \quad &\text{otherwise.}
\end{cases} \]
for $x, y \in V$.
\end{lem}

\begin{proof}
The same statement appears in~\cite{segawa2021quantum, yoshie2022quantum}, where no proof was provided.
For completeness, we present a proof here.
For $x, y \in V$, we have
\begin{align*}
(T_{\sigma})_{x,y}
&= \sum_{a,b \in \MC{A}} (d_{\sigma})_{x,a} S_{a,b} (d_{\sigma}^*)_{b,y} \\
&= \sum_{a,b \in \MC{A}} \frac{\sigma(a)}{\sqrt{\deg x}} \delta_{x, t(a)} \cdot \delta_{a, b^{-1}} \cdot \frac{\sigma(b)}{\sqrt{\deg y}} \delta_{y, t(b)} \\
&= \frac{1}{\sqrt{\deg x} \sqrt{\deg y}} \sum_{a \in \MC{A}} \sigma(a) \sigma(a^{-1}) \delta_{x, t(a)} \delta_{y, t(a^{-1})} \\
&= \frac{1}{\sqrt{\deg x} \sqrt{\deg y}} \sum_{a \in \MC{A}} \sigma(a) \sigma(a^{-1}) \delta_{x, t(a)} \delta_{y, o(a)}.
\end{align*}
If $\{x,y\} \in E$,
then we can take $a=(y,x)$ as the unique arc satisfying $x = t(a)$ and $y = o(a)$.
Hence
\[ (T_{\sigma})_{x,y} = \frac{\sigma(x,y)\sigma(y,x)}{\sqrt{\deg x} \sqrt{\deg y}}. \]
If $\{x,y\} \not\in E$,
then no arc satisfies $x = t(a)$ and $y = o(a)$,
and thus $(T_{\sigma})_{x,y} = 0$.
\end{proof}

We search for a marked arc in the complete bipartite graph $K_{n,n}$. %
Since $K_{n,n}$ is arc-transitive,
we may assume that an arbitrary arc, say $a = (x',y') \in \MC{A}(K_{n,n})$, is marked without loss of generality.
Considering the quantum search model determined by the sign function $\sigma_a$,
Lemma~\ref{1016-1} implies that the corresponding discriminant matrix $T_{\sigma_a}$ is the normalized adjacency matrix of an edge-signed graph with precisely one negative edge.
That is, $T_{\sigma_a}$ coincides with the normalized adjacency matrix of $K_{n,n}$,
except that the two entries $(x',y')$ and $(y',x')$ are replaced by $-1/n$ instead of $+1/n$.

The eigenvalues and eigenvectors of such matrices have been studied in spectral graph theory.
We present the statement in a form adapted to our setting,
while the essential result is due to Akbari et al.\cite{akbari2018spectrum}.

\begin{thm}[{\cite[Theorem~4.1]{akbari2018spectrum}}] \label{0908-3}
Let $n \geq 2$, and let $a=(x',y')$ be an arc of $\MC{A}(K_{n,n})$.
The spectrum of $T_{\sigma_a}$ is
\[ \Spec(T_{\sigma_a}) = \{ [\pm \lambda]^1, [\pm \mu]^1, [0]^{2n-4} \}, \]
where $\lambda$ and $\mu$ are as follows:
\[ \lambda = \frac{n-2 + \sqrt{\D}}{2n},
\quad \mu = \frac{-(n-2) + \sqrt{\D}}{2n},
\quad \D = n^2 + 4n -4.
\]
The normalized eigenvector $\BM{f}$ associated with the largest eigenvalue $\lambda$ is
\begin{equation} \label{0908-1}
\BM{f}_x = \frac{1}{\sqrt{\D - n \sqrt{\D}}} \cdot
\begin{cases}
\frac{-n + \sqrt{\D}}{2} \quad &\text{if $x \in \{x', y'\}$,} \\
1 \quad &\text{otherwise.}
\end{cases}
\end{equation}
\end{thm}

Note that Akbari et al.~\cite{akbari2018spectrum} determined both the eigenvalues and the eigenvectors of edge-signed complete bipartite graphs in which the negative edges form a matching.
Among their results, we are concerned only with the largest eigenvalue $\lambda$ and its corresponding eigenvector $\BM{f}$.
Further research on the eigenvalue analysis of edge-signed complete bipartite graphs has been conducted by Pirzada et al.~\cite{pirzada2021eigenvalues}.

We remark that if $n \geq 3$,
then $(n+1)^2 < \D < (n+2)^2$.

\begin{lem} \label{0908-2}
Let $\BM{f} \in \MB{C}^{V(K_{n,n})}$ be the vector defined by Equality~\eqref{0908-1}.
For an arc $z \in \MC{A}(K_{n,n})$, we have
\begin{equation}
(d_{\sigma_a}^* \BM{f})_z = \frac{\sigma_a(z)}{\sqrt{n}} \BM{f}_{t(z)},
\end{equation}
and
\begin{equation}
(S d_{\sigma_a}^* \BM{f})_z = \frac{\sigma_a(z^{-1})}{\sqrt{n}} \BM{f}_{o(z)}.
\end{equation}
\end{lem}

\begin{proof}
For convenience, let $V := V(K_{n,n})$ and $\MC{A} := \MC{A}(K_{n,n})$.
Then we have
\[ (d_{\sigma_a}^* \BM{f})_z
= \sum_{x \in V} (d_{\sigma_a}^*)_{z,x} \BM{f}_x
= \sum_{x \in V} \frac{\sigma_a(z)}{\sqrt{n}}\delta_{x, t(z)} \BM{f}_x
= \frac{\sigma_a(z)}{\sqrt{n}}\BM{f}_{t(z)}.
\]
This in turn yields the second relation:
\begin{align*}
(S d_{\sigma_a}^* \BM{f})_z
&= \sum_{w \in \MC{A}} S_{z,w} (d_{\sigma_a}^* \BM{f})_w
= \sum_{w \in \MC{A}} \delta_{z, w^{-1}} (d_{\sigma_a}^* \BM{f})_w \\
&= (d_{\sigma_a}^* \BM{f})_{z^{-1}}
= \frac{\sigma_a(z^{-1})}{\sqrt{n}}\BM{f}_{t(z^{-1})}
= \frac{\sigma_a(z^{-1})}{\sqrt{n}} \BM{f}_{o(z)},
\end{align*}
as claimed.
\end{proof}

Associated with the largest eigenvalue $\lambda$ of $T_{\sigma_a}$
and its normalized eigenvector $\BM{f}$,
we define and shall use the following notation throughout this section:
\begin{align*}
\theta &:= \arccos \lambda, \\
\BM{\varphi}_{\pm} &:= \frac{1}{\sqrt{2} \sin \theta}(d_{\sigma_a}^* - e^{\pm i \theta} S d_{\sigma_a}^*) \BM{f}, \\
\BM{\beta}_{\pm} &:= \frac{1}{\sqrt{2}}(\BM{\varphi}_+ \pm \BM{\varphi}_-), \\
t^* &:= \left \lfloor \frac{\pi}{2\theta} \right \rfloor, \\
p^* &:= | (U_{\sigma_a}^{t^*} \init)_a |^2.
\end{align*}

Here, $t^*$ is the measurement time,
and $p^*$ is the probability of finding the marked arc $a$ at time $t^*$.
The justification for taking $t^*$ as an appropriate measurement time is provided in~\cite[Remark~4.3]{yoshie2022quantum}.
Roughly speaking, if the largest eigenvalue $\lambda$ of $T_{\sigma_a}$ is sufficiently close to $1$,
then we have $p^* = | (U_{\sigma_a}^{t^*} \init)_a |^2 \approx | (\BM{\beta}_{+})_a |^2$.
In fact, $\BM{\beta}_+$ is a vector whose probability amplitude is concentrated on $a$ and $a^{-1}$,
and hence we will see that $|(U_{\sigma_a}^{t^*} \init)_a|^2$ is approximately $1/2$.
In what follows, we provide mathematically rigorous arguments to justify the above rough discussion.

\subsection{Technical estimates for the success probability $p^*$}

In this subsection, we present technical estimates needed to evaluate the success probability $p^*$.
We prove that the quantity $\sqrt{p^*}$ is decomposed into three terms,
and estimate these three terms separately.

\begin{lem} \label{0905-1}
With the above notation, we have
\begin{equation} \label{0909-1}
\sqrt{p^*} = | (U_{\sigma_a}^{t^*} \init)_a |
\geq |(\BM{\beta}_+)_a| - \| U_{\sigma_a}^{t^*}(i \BM{\beta}_-) + \BM{\beta}_+ \| - \| \init - i \BM{\beta}_- \|.
\end{equation}
\end{lem}

\begin{proof}
The proof essentially reduces to an application of the triangle inequality.
First,
\begin{align*}
| (U_{\sigma_a}^{t^*} \init)_a |
&= |-(\BM{\beta}_+)_a + (U_{\sigma_a}^{t^*} \init)_a + (\BM{\beta}_+)_a| \\
&\geq | -(\BM{\beta}_+)_a | - | (U_{\sigma_a}^{t^*} \init)_a + (\BM{\beta}_+)_a | \\
&\geq | (\BM{\beta}_+)_a | - \| U_{\sigma_a}^{t^*} \init + \BM{\beta}_+ \|.
\end{align*}
On the other hand,
\begin{align*}
\| U_{\sigma_a}^{t^*} \init + \BM{\beta}_+ \|
&= \| U_{\sigma_a}^{t^*} \init - U_{\sigma_a}^{t^*}(i \BM{\beta}_-) + U_{\sigma_a}^{t^*}(i \BM{\beta}_-)  + \BM{\beta}_+ \| \\
&\leq \| U_{\sigma_a}^{t^*} \init - U_{\sigma_a}^{t^*}(i \BM{\beta}_-) \| + \| U_{\sigma_a}^{t^*}(i \BM{\beta}_-)  + \BM{\beta}_+ \| \\
&= \| \init - i \BM{\beta}_- \| + \| U_{\sigma_a}^{t^*}(i \BM{\beta}_-)  + \BM{\beta}_+ \|.
\end{align*}
Therefore,
\begin{align*}
| (U_{\sigma_a}^{t^*} \init)_a |
&\geq | (\BM{\beta}_+)_a | - \left( \| \init - i \BM{\beta}_- \| + \| U_{\sigma_a}^{t^*}(i \BM{\beta}_-)  + \BM{\beta}_+ \| \right) \\
&= | (\BM{\beta}_+)_a | - \| \init - i \BM{\beta}_- \| - \| U_{\sigma_a}^{t^*}(i \BM{\beta}_-)  + \BM{\beta}_+ \|,
\end{align*}
as claimed.
\end{proof}

Hence, evaluating the success probability $p^*$ reduces to evaluating the three terms on the right-hand side of Inequality~\eqref{0909-1}.

\begin{pro} \label{0917-1}
We have
\begin{equation} \label{0905-2}
(\BM{\beta}_+)_a=-\sqrt{\frac{n-1}{4n}\left(1+\frac{n}{\sqrt{\D}}\right)}.
\end{equation}
In particular, for every $n\geq2$,
\begin{equation} \label{0905-3}
|(\BM{\beta}_+)_a|>\frac{n-1}{\sqrt2\,n}.
\end{equation}
\end{pro}

\begin{proof}
By Lemma~\ref{0908-2}, we have
\[ (d_{\sigma_a}^* \BM{f})_a = \frac{1}{\sqrt{n}} \cdot \frac{n - \sqrt{\D}}{2 \sqrt{\D - n \sqrt{\D}}}, \qquad
(Sd_{\sigma_a}^* \BM{f})_a = \frac{1}{\sqrt{n}} \cdot \frac{-n + \sqrt{\D}}{2 \sqrt{\D - n \sqrt{\D}}}.
\]
Thus,
\begin{align}
(\BM{\beta}_+)_a
&=\frac{1}{\sqrt2}\left\{(\BM{\varphi}_+)_a+(\BM{\varphi}_-)_a\right\} \notag \\
&=\frac{1}{2\sin\theta}\left\{(d_{\sigma_a}^*\BM f)_a-e^{i\theta}(Sd_{\sigma_a}^*\BM f)_a
+(d_{\sigma_a}^*\BM f)_a-e^{-i\theta}(Sd_{\sigma_a}^*\BM f)_a\right\} \notag \\
&=\frac{(d_{\sigma_a}^*\BM f)_a-\lambda(Sd_{\sigma_a}^*\BM f)_a}{\sin\theta} \notag \\
&=\frac{n-\sqrt{\D}}{2\sqrt{n}\sqrt{\D-n\sqrt{\D}}}\cdot\frac{1+\lambda}{\sin\theta}. \label{1118-1}
\end{align}
By Theorem~\ref{0908-3}, noting that $\D=n^2+4n-4$ and $\sin\theta=\sqrt{1-\lambda^2}>0$, we have
\begin{align*}
\frac{1+\lambda}{\sin\theta}
&=\sqrt{\frac{1+\lambda}{1-\lambda}}
=\sqrt{\frac{3n-2+\sqrt{\D}}{n+2-\sqrt{\D}}} \\
&=\sqrt{\frac{(3n-2+\sqrt{\D})(n+2+\sqrt{\D})}{(n+2)^2-\D}} \\
&=\sqrt{\frac{(n+\sqrt{\D})^2}4}
=\frac{n+\sqrt{\D}}2.
\end{align*}
Hence, noting that $\D-n^2=4(n-1)$ and $\sqrt{\D}>n$, we obtain
\begin{align*}
(\BM{\beta}_+)_a
&=\frac{n-\sqrt{\D}}{2\sqrt n\sqrt{\D-n\sqrt{\D}}}\cdot\frac{n+\sqrt{\D}}2 \\*
&=-\frac{n+\sqrt{\D}}{4\sqrt n}\sqrt{\frac{\sqrt{\D}-n}{\sqrt{\D}}} \\*
&=-\sqrt{\frac{(n+\sqrt{\D})(\D-n^2)}{16n\sqrt{\D}}} \\*
&=-\sqrt{\frac{n-1}{4n}\left(1+\frac{n}{\sqrt{\D}}\right)},
\end{align*}
as claimed.
Next, we show Inequality~\eqref{0905-3}.
Since $n^4-(n-2)^2\D=16(n-1)^2>0$, we have $n/\sqrt{\D}>(n-2)/n$ for $n\geq2$.
Thus, by~\eqref{0905-2},
\[ |(\BM{\beta}_+)_a|>\sqrt{\frac{n-1}{4n}\left(1+\frac{n-2}{n}\right)}
=\frac{n-1}{\sqrt2\,n}, \]
as claimed.
\end{proof}

Next, we estimate an upper bound for the second term appearing on the right-hand side of Inequality~\eqref{0909-1}.
For later use, we state the following elementary identity:
\begin{equation} \label{0909-4}
|i e^{i \phi} + 1|^2 = 2 - 2 \sin \phi
\end{equation}
for a real number $\phi$.

\begin{pro} \label{0917-2}
We have
\begin{equation} \label{0909-2}
\| U_{\sigma_a}^{t^*}(i \BM{\beta}_-) + \BM{\beta}_+ \|
\leq \sqrt{\frac{n+2 - \sqrt{\D}}{n}}.
\end{equation}
In particular, for sufficiently large $n$,
\begin{equation} \label{0909-3}
\| U_{\sigma_a}^{t^*}(i \BM{\beta}_-) + \BM{\beta}_+ \|
< \sqrt{\frac{8}{n(2n+3)}}.
\end{equation}
\end{pro}

\begin{proof}
First, we have
\begin{align*}
\| U_{\sigma_a}^{t^*}(i \BM{\beta}_-) + \BM{\beta}_+ \|^2
&= \left \| U_{\sigma_a}^{t^*} \left \{ \frac{i}{\sqrt{2}}(\BM{\varphi}_+ - \BM{\varphi}_-) \right \} + \frac{1}{\sqrt{2}}(\BM{\varphi}_+ + \BM{\varphi}_-) \right \|^2 \\
&= \left \|
\frac{i}{\sqrt{2}}( e^{i t^* \theta}\BM{\varphi}_+ - e^{-i t^*\theta} \BM{\varphi}_- ) + \frac{1}{\sqrt{2}}(\BM{\varphi}_+ + \BM{\varphi}_-) \right \|^2 \\
&= \left \|
\left( \frac{i}{\sqrt{2}} e^{i t^* \theta} + \frac{1}{\sqrt{2}} \right)\BM{\varphi}_+ +
\left( -\frac{i}{\sqrt{2}} e^{-i t^* \theta} + \frac{1}{\sqrt{2}} \right)\BM{\varphi}_- \right \|^2 \\
&\leq \left \|
\frac{1}{\sqrt{2}} \left( i e^{i t^* \theta} + 1 \right)\BM{\varphi}_+ \right \|^2 + \left \|
\frac{1}{\sqrt{2}} \left( -i e^{-i t^* \theta} + 1 \right)\BM{\varphi}_- \right \|^2 \\
&= \frac{1}{2} \left| i e^{i t^*\theta} + 1 \right|^2 + \frac{1}{2} \left| -i e^{-i t^*\theta} + 1 \right|^2.
\end{align*}
Since $-i e^{-i t^*\theta} + 1$ is the complex conjugate of $i e^{i t^*\theta} + 1$, their absolute values are equal.
Thus,
\[ \frac{1}{2} \left| i e^{i t^*\theta} + 1 \right|^2 + \frac{1}{2} \left| -i e^{-i t^*\theta} + 1 \right|^2 = |i e^{i t^*\theta} + 1|^2. \]
By Equality~\eqref{0909-4}, we obtain
\[ |i e^{i t^*\theta} + 1|^2
= 2(1-\sin t^*\theta) = 2 \left\{ 1- \cos \left( \frac{\pi}{2} - t^*\theta \right) \right \} = 2 \left\{ 1- \cos \left( \frac{\pi}{2} -  \left \lfloor \frac{\pi}{2\theta} \right \rfloor \theta \right) \right \}. \]
Here, $0 < \theta < \frac{\pi}{2}$ and $\frac{\pi}{2\theta} - 1 \leq \lfloor \frac{\pi}{2\theta} \rfloor \leq \frac{\pi}{2\theta}$ imply
\[ 0 \leq \frac{\pi}{2} - \left \lfloor \frac{\pi}{2\theta} \right \rfloor \theta \leq \theta < \frac{\pi}{2}. \]
Hence,
\[ \| U_{\sigma_a}^{t^*}(i \BM{\beta}_-) + \BM{\beta}_+ \|^2
\leq 2(1-\cos \theta) = 2(1 - \lambda) = \frac{n+2 - \sqrt{\D}}{n}, \]
where the last equality follows from Theorem~\ref{0908-3}.

We proceed to show Inequality~\eqref{0909-3}.
We have
\begin{align*}
\| U_{\sigma_a}^{t^*}(i \BM{\beta}_-) + \BM{\beta}_+ \|
&\leq \sqrt{\frac{n+2 - \sqrt{\D}}{n}} \tag{by \eqref{0909-2}} \\
&= \sqrt{\frac{(n+2)^2 - \D}{n(n+2 + \sqrt{\D})}} \\
&< \sqrt{\frac{8}{n (n+2 + (n+1))}} \\
&= \sqrt{\frac{8}{n(2n+3)}},
\end{align*}
as claimed.
\end{proof}

Finally, we estimate an upper bound for the third term.

\begin{pro} \label{0917-3}
We have
\begin{equation*} %
\| \init - i \BM{\beta_-} \|
= \sqrt{2 - \frac{\sqrt{2}(n-1)(\sqrt{\D}+n)}{n \sqrt{n} \sqrt{\D - n\sqrt{\D}} }}.
\end{equation*}
In particular, for every $n\geq2$,
\begin{equation} \label{0910-2}
\|\init-i\BM{\beta}_-\|<\frac{\sqrt2}{n}.
\end{equation}
\end{pro}

\begin{proof}
First, we confirm that
\begin{equation} \label{0910-3}
i \BM{\beta_-} = Sd_{\sigma_a}^* \BM{f}
\end{equation}
and that $i \BM{\beta_-} \in \MB{R}^{\MC{A}}$.
Indeed,
\begin{align*}
\BM{\beta_-}
&= \frac{1}{\sqrt{2}}(\BM{\varphi}_+ - \BM{\varphi}_-) \\
&= \frac{1}{\sqrt{2}} \cdot \frac{1}{\sqrt{2} \sin \theta} \left\{ (d_{\sigma_a}^* - e^{i \theta}Sd_{\sigma_a}^*) - (d_{\sigma_a}^* - e^{-i \theta}Sd_{\sigma_a}^*) \right\} \BM{f} \\
&= \frac{1}{\sqrt{2}} \cdot \frac{-2i \sin \theta}{\sqrt{2} \sin \theta} Sd_{\sigma_a}^*\BM{f} \\
&= -i Sd_{\sigma_a}^*\BM{f}.
\end{align*}
Hence,
\begin{align*}
\| \init - i \BM{\beta_-} \|^2
&= (\init - i \BM{\beta_-}, \init - i \BM{\beta_-}) \\
&= \| \init \|^2 + \| i \BM{\beta_-} \|^2 - (\init, i \BM{\beta_-}) - (i \BM{\beta_-}, \init) \\
&= 2 - 2(i \BM{\beta_-}, \init) \\
&= 2 - 2(Sd_{\sigma_a}^* \BM{f}, \init). \tag{by~\eqref{0910-3}}
\end{align*}
By Lemma~\ref{0908-2}, we obtain
\begin{align*}
(Sd_{\sigma_a}^*\BM{f})_z
&= \frac{\sigma_a(z^{-1})}{\sqrt{n}} \BM{f}_{o(z)} \\
&= \frac{1}{\sqrt{n}\sqrt{\D - n\sqrt{\D}}} \cdot \begin{cases}
- \frac{\sqrt{\D} - n}{2} \quad \quad &\text{if $z = a^{-1}$,} \\
\frac{\sqrt{\D} - n}{2} \quad &\text{if $z \neq a^{-1}$ and $o(z) \in \{o(a), t(a)\}$,} \\
1 \quad &\text{otherwise.}
\end{cases}
\end{align*}
The number of arcs $z$ satisfying each of the above conditions is, from top to bottom, $1$, $2n-1$, and $2n(n-1)$, respectively.
Thus,
\begin{align*}
(Sd_{\sigma_a}^* \BM{f}, \init)
&= \frac{1}{\sqrt{2n^2}} \cdot \frac{1}{\sqrt{n}\sqrt{\D - n\sqrt{\D}}}
\left( -\frac{\sqrt{\D} - n}{2} + (2n-1) \cdot \frac{\sqrt{\D} - n}{2} + 2n(n-1) \cdot 1 \right) \\
&= \frac{(n-1)(\sqrt{\D} + n)}{\sqrt{2} \cdot n\sqrt{n}\sqrt{\D - n\sqrt{\D}}}.
\end{align*}
Therefore, we have
\[ \| \init - i \BM{\beta_-} \|^2
= 2 - 2 \cdot \frac{(n-1)(\sqrt{\D} + n)}{\sqrt{2} \cdot n\sqrt{n}\sqrt{\D - n\sqrt{\D}}}
= 2 - \frac{\sqrt{2}(n-1)(\sqrt{\D}+n)}{n \sqrt{n} \sqrt{\D - n\sqrt{\D}} }. \]

\begin{samepage}
We proceed to show Inequality~\eqref{0910-2}.
Noting that $\D-n^2=4(n-1)$ and $n<\sqrt\D<n+2$ for $n\geq2$, we have
\begin{align*}
(Sd_{\sigma_a}^*\BM f,\init)^2-\left(1-\frac1{n^2}\right)^2
&=\frac{(n-1)^2(\sqrt\D+n)^2}{2n^3(\D-n\sqrt\D)}-\left(1-\frac1{n^2}\right)^2 \\*
&=\frac{(n-1)(\sqrt\D+n)^3}{8n^3\sqrt\D}-\left(1-\frac1{n^2}\right)^2 \\*
&=\frac{n-1}{8n^4\sqrt\D}\left\{8\sqrt\D-n(\sqrt\D-n)^3\right\} \\*
&>\frac{n-1}{8n^4\sqrt\D}(8\sqrt\D-8n)>0.
\end{align*}
\end{samepage}
Since $(Sd_{\sigma_a}^*\BM f,\init)>0$ and $1-1/n^2>0$, it follows that
\[ (Sd_{\sigma_a}^*\BM f,\init)>1-\frac1{n^2}. \]
Hence,
\[ \|\init-i\BM{\beta}_-\|^2
=2-2(Sd_{\sigma_a}^*\BM f,\init)<\frac2{n^2}, \]
as claimed.
\end{proof}

We now discuss Inequality~\eqref{0909-1} derived in Lemma~\ref{0905-1}.
In quantum search algorithms such as Grover's algorithm, the time evolution of the quantum state is often interpreted in terms of a ``rotation."
A similar structure can also be found in our setting.
Indeed, by a simple calculation, we find that
\[ U_{\sigma_a} \BM{\beta_+} = \cos \theta \BM{\beta_+} + i \sin \theta \BM{\beta_-}, \qquad U_{\sigma_a} \BM{\beta_-} = i \sin \theta \BM{\beta_+} + \cos \theta \BM{\beta_-} \]
and therefore obtain
\[ U_{\sigma_a} [\BM{\beta_+}, \BM{\beta_-}] =
[\BM{\beta_+}, \BM{\beta_-}]
\begin{bmatrix} 
\cos\theta&i\sin\theta \\ 
i\sin\theta&\cos\theta
\end{bmatrix}. \]
Thus, let
\[ \Lambda := \begin{bmatrix} 
\cos\theta&i\sin\theta \\ 
i\sin\theta&\cos\theta
\end{bmatrix}. \]
Then, for any positive integer $\tau$,
\[
\Lambda^{\tau} = \begin{bmatrix} 
\cos (\tau\theta)&i\sin (\tau\theta) \\ 
i\sin(\tau\theta)&\cos(\tau\theta)
\end{bmatrix}. \]
Therefore, we obtain $U_{\sigma_a}^{\tau} [\BM{\beta_+}, \, \BM{\beta_-}] = [\BM{\beta_+}, \, \BM{\beta_-}] \Lambda^{\tau}$.
In particular,
\[
U_{\sigma_a}^{\tau} \BM{\beta_-} = i \sin (\tau \theta) \BM{\beta_+} + \cos (\tau \theta) \BM{\beta_-}.
\]
Hence, the state at time $\tau$ can be written as
\begin{align*}
U_{\sigma_a}^{\tau} \init
&= U_{\sigma_a}^{\tau} (i \BM{\beta}_- + \init - i \BM{\beta}_- ) \\
&= - \BM{\beta}_+ + \{ U_{\sigma_a}^{\tau} (i \BM{\beta}_-) + \BM{\beta}_+ \} + U_{\sigma_a}^{\tau} (\init - i \BM{\beta}_- ).
\end{align*}
Since the second term is
\[ U_{\sigma_a}^{\tau} (i \BM{\beta}_-) + \BM{\beta}_+ = (1-\sin (\tau \theta)) \BM{\beta}_+ + i \cos (\tau \theta) \BM{\beta}_-, \]
it is close to the zero vector at a time $\tau$ satisfying $\tau \theta \approx \frac{\pi}{2}$, namely, at $\tau = t^*$.
On the other hand, we have $i \BM{\beta}_- = Sd^*_{\sigma_a} \BM{f}$ by Equality~\eqref{0910-3}.
In our model, since $\lambda$ is close to $1$ and the effect of marking is small and local, the corresponding eigenvector $\BM{f}$ is expected not to differ significantly from the eigenvector corresponding to the largest eigenvalue $1$ of the ordinary unsigned normalized adjacency matrix.
From this, $i\BM{\beta}_-$ is expected to be close to the initial state $\init$.
However, it is not clear whether this expectation is borne out for general graphs.
By Proposition~\ref{0917-3}, at least in the case of $K_{n,n}$, we see that $i \BM{\beta}_-$ is close to $\init$, and consequently $U_{\sigma_a}^{t^*} \init$ is close to $-\BM{\beta_+}$.

\subsection{A brief review of Big-Theta notation}

This subsection briefly reviews the definition and basic properties of Big-Theta notation,
which is used to describe the asymptotic order of functions.
Throughout this paper,
we denote the set of positive integers by $\MB{N}$.
Let $f, g: \MB{N} \to \MB{R}$ be functions with $f(n), g(n) > 0$ for all sufficiently large $n$.
We say that $f$ and $g$ are {\it of the same order},
written as $f(n) = \Theta(g(n))$,
if there exist a positive integer $N$ and positive constants $c_1, c_2$ such that the implication
\[ n \geq N \quad \Longrightarrow \quad c_1g(n) \leq f(n) \leq c_2g(n) \]
holds.
If the quotient $f/g$ tends to a positive finite limit,
then the two functions are of the same order:

\begin{lem} \label{0912-2}
Let $f, g: \MB{N} \to \MB{R}$ be functions with $f(n), g(n) > 0$ for all sufficiently large $n$.
If
\[ \lim_{n \to \infty} \frac{f(n)}{g(n)} = c \, (>0) \]
then $f(n) = \Theta(g(n))$.
\end{lem}

\begin{proof}
The limit $c$ exists and is positive,
and hence let $\varepsilon := c/2 > 0$.
Then there exists a positive integer $N$ such that for all $n \geq N$, 
\[ \left| \frac{f(n)}{g(n)} - c \right| < \varepsilon = \frac{c}{2}. \]
This inequality implies
\[ \frac{c}{2}g(n) < f(n) < \frac{3c}{2}g(n), \]
and therefore the claim follows.
\end{proof}

\subsection{The order of the measurement time $t^*$ and an asymptotic estimate of the success probability $p^*$}
In this subsection, we derive an asymptotic estimate of the measurement time $t^*$ and the success probability $p^*$.

\begin{lem} \label{0912-1}
For $x \in (0,1)$, we have
\[ \frac{1}{\sqrt{2(1-x)}} -1 \leq \frac{1}{\arccos x} \leq \frac{1}{\sqrt{2(1-x)}}. \]
\end{lem}

\begin{thm}
We have
\[ \frac{\pi}{2} \left( \sqrt{\frac{n(2n+3)}{8}} - 2 \right)
\leq t^* \leq \frac{\pi}{4}\sqrt{n(n+2)}. \]
In particular, $t^* = \Theta(n)$.
\end{thm}

\begin{proof}
First, we establish the upper bound.
We have
\begin{align*}
t^*
&\leq \frac{\pi}{2 \arccos \lambda} \\
&\leq \frac{\pi}{2} \cdot \frac{1}{\sqrt{2(1-\lambda)}} \tag{by Lemma~\ref{0912-1}} \\
&= \frac{\pi}{2} \cdot \sqrt{\frac{n}{n+2 - \sqrt{\D}}} \tag{by Theorem~\ref{0908-3}} \\
&= \frac{\pi}{2} \cdot \sqrt{\frac{n(n+2 + \sqrt{\D})}{(n+2 - \sqrt{\D})(n+2 + \sqrt{\D})}} \\
&\leq \frac{\pi}{2} \cdot \sqrt{\frac{n(n+2 + n+2)}{8}} \\
&= \frac{\pi}{4}\sqrt{n(n+2)}.
\end{align*}
For $n=2$, the lower bound is immediate from $t^*=2$. Assume $n\geq3$.
The lower bound can be proven as follows.
\begin{align*}
t^* &\geq \frac{\pi}{2 \arccos \lambda} - 1 \\
&\geq \frac{\pi}{2}\left( \frac{1}{\sqrt{2(1-\lambda)}} - 1 \right) - \frac{\pi}{2} \tag{by Lemma~\ref{0912-1}} \\
&= \frac{\pi}{2}\left( \sqrt{\frac{n}{n+2-\sqrt{\D}}} - 2 \right) \tag{by Theorem~\ref{0908-3}} \\
&= \frac{\pi}{2}\left( \sqrt{\frac{n(n+2+\sqrt{\D})}{(n+2-\sqrt{\D})(n+2+\sqrt{\D})}} - 2 \right) \\
&\geq \frac{\pi}{2}\left( \sqrt{\frac{n(n+2+n+1)}{8}} - 2 \right) \\
&= \frac{\pi}{2} \left( \sqrt{\frac{n(2n+3)}{8}} - 2 \right).
\end{align*}
Furthermore, dividing every term of the obtained inequality by $n$ yields
\[ \frac{\pi}{2} \left( \sqrt{\frac{1}{4}+\frac{3}{8n}} - \frac{2}{n} \right) \leq \frac{t^*}{n} \leq \frac{\pi}{4} \sqrt{1+\frac{2}{n}}. \]
Since the leftmost and rightmost terms converge to $\frac{\pi}{4}$ as $n \to \infty$,
Lemma~\ref{0912-2} implies that $t^* = \Theta(n)$.
\end{proof}

Note that the number of arcs is $|\MC{A}(K_{n,n})| = 2n^2 = \Theta(n^2)$, and thus a quadratic speedup is achieved.
Next, we derive an asymptotic estimate of $p^*$.

\begin{thm} \label{0930-1}
For every integer $n\geq6$, we have
\begin{equation}\label{v2-polynomial-bound}
p^*\geq P_n:=\frac12-\frac{3+2\sqrt2}{n}+\frac{17+12\sqrt2}{2n^2}.
\end{equation}
In particular, $p^*\geq\frac12-\frac{3+2\sqrt2}{n}$ for every $n\geq2$.
\end{thm}

\begin{proof}
First, {by Lemma~\ref{0905-1}, we have}
\begin{equation} \label{0930-2}
\sqrt{p^*} \geq |(\BM{\beta}_+)_a| - \| U_{\sigma_a}^{t^*}(i \BM{\beta}_-) + \BM{\beta}_+ \| - \| \init - i \BM{\beta}_- \| 
\end{equation}
Noting that $(n+2)^2-\D=8$ and $\sqrt{\D}>n$, we have
\[ \sqrt{\frac{n+2-\sqrt\D}{n}}
=\sqrt{\frac8{n(n+2+\sqrt\D)}}<\frac2n. \]
By Inequality~\eqref{0930-2} and Propositions~\ref{0917-1},~\ref{0917-2}, and~\ref{0917-3}, we obtain
\begin{align*}
\sqrt{p^*}
&\geq\frac{n-1}{\sqrt2\,n}-\sqrt{\frac{n+2-\sqrt{\D}}n}-\frac{\sqrt2}n \\
&>\frac{n-1}{\sqrt2\,n}-\frac2n-\frac{\sqrt2}{n} \\
&=\frac1{\sqrt2}\left(1-\frac{3+2\sqrt2}{n}\right).
\end{align*}
For $n\geq6$, the right-hand side is non-negative, and hence squaring both sides yields
\begin{align*}
p^*&\geq\frac12\left(1-\frac{3+2\sqrt2}{n}\right)^2 \\
&=\frac12-\frac{3+2\sqrt2}{n}+\frac{17+12\sqrt2}{2n^2}=P_n.
\end{align*}
The last assertion follows from this bound for $n\geq6$ and from $p^*\geq0$ for $2\leq n\leq5$.
\end{proof}

From this theorem,
we see that the lower bound of the success probability $p^*$ converges to $1/2$ at a rate of $1/n$.

In previous studies~\cite{segawa2021quantum, yoshie2022quantum} on edge search using the same model, the corresponding classical search problems can also be formulated in a natural way, and comparisons with quantum search are therefore discussed.
On the other hand, in the arc search treated in this paper, the coordinate of the quantum state corresponding to a marked arc is interpreted as a pair consisting of the position of the walker and its internal state.
Since no classical search problem directly corresponds to such a setting, this paper focuses solely on the analysis of quantum walks.

\section{{The limit of the success probability}}

In this paper, we {have proved} that the limiting value of the lower bound on the success probability $p^*$ for the search on the complete bipartite graphs $K_{n,n}$ is $1/2$ in time $\Theta(n)$.
In the remainder of this paper, we discuss the reason why the value $1/2$ appears.
{Moreover, this discussion implies that the success probability $p^*$ actually converges to $1/2$.}

It can be shown that the second and third terms on the right-hand side of Inequality~\eqref{0909-1} satisfy
\begin{equation}\label{0924-1}
\lim_{n \to \infty} \| U_{\sigma_a}^{t^*}(i \BM{\beta}_-) + \BM{\beta}_+ \| = \lim_{n \to \infty} \| \init - i \BM{\beta}_- \| = 0
\end{equation}
as a consequence of Propositions~\ref{0917-2} and~\ref{0917-3}.
Therefore, the lower bound can be regarded as being primarily determined by the vector $\BM{\beta_+}$.
Furthermore, it can be seen that the amplitude is concentrated at the arcs $a$ and $a^{-1}$ when $n$ is sufficiently large, as shown below:

\begin{lem} \label{1118-2}
We have
\[ (\BM{\beta_+})_{a^{-1}} = - (\BM{\beta_+})_{a}. \]
\end{lem}
\begin{proof}
Indeed, Lemma~\ref{0908-2} and Theorem~\ref{0908-3} yield
\begin{align*}
(d^*_{\sigma_{a}}\BM{f})_{a^{-1}}
&= \frac{\sigma_{a}(a^{-1})}{\sqrt{n}} \BM{f}_{t(a^{-1})}
= \frac{1}{\sqrt{n}} \cdot \frac{1}{\sqrt{\Delta - n \sqrt{\Delta}}} \cdot \frac{\sqrt{\Delta} - n}{2}, \\
(Sd^*_{\sigma_{a}}\BM{f})_{a^{-1}}
&= \frac{\sigma_{a}(a)}{\sqrt{n}} \BM{f}_{o(a^{-1})}
= - \frac{1}{\sqrt{n}} \cdot \frac{1}{\sqrt{\Delta - n \sqrt{\Delta}}} \cdot \frac{\sqrt{\Delta} - n}{2},
\end{align*}
and hence
\begin{align*}
(\BM{\varphi}_{\pm})_{a^{-1}}
&= \frac{1}{\sqrt{2}\sin \theta} \left\{ (d^*_{\sigma_{a}}\BM{f})_{a^{-1}} - e^{\pm i \theta} (Sd^*_{\sigma_{a}}\BM{f})_{a^{-1}} \right\} \\
&= \frac{1}{\sqrt{2}\sin \theta} \cdot \frac{\sqrt{\Delta} - n}{2\sqrt{n} \sqrt{\Delta - n \sqrt{\Delta}}}(1+e^{\pm i \theta}).
\end{align*}
Therefore,
\begin{align*}
(\BM{\beta_+})_{a^{-1}}
&= \frac{1}{\sqrt{2}} \left( (\BM{\varphi}_{+})_{a^{-1}} + (\BM{\varphi}_{-})_{a^{-1}} \right) \\
&= \frac{\sqrt{\Delta} - n}{2 \sin \theta \cdot 2\sqrt{n} \sqrt{\Delta - n \sqrt{\Delta}}}(2+2\cos \theta) \\
&= - (\BM{\beta_+})_{a}, \tag{by \eqref{1118-1}}
\end{align*}
as claimed.
\end{proof}

By the above lemma, we can obtain exactly the same result as Theorem~\ref{0930-1} for the inverse arc of the marked arc.
Let
\[ q^* := | (U_{\sigma_a}^{t^*} \init)_{a^{-1}} |^2. \]

\begin{lem} \label{0930-3}
For every integer $n\geq6$, we have $q^*\geq P_n$.
In particular, $q^*\geq\frac12-\frac{3+2\sqrt2}{n}$ for every $n\geq2$.
\end{lem}

\begin{proof}
By Lemma~\ref{1118-2} and the triangle inequality, we have
\begin{align*}
|(\BM{\beta}_{+})_a| &= |(\BM{\beta}_{+})_{a^{-1}}| \\
&= |(\BM{\beta}_{+})_{a^{-1}} + (U_{\sigma_a}^{t^*} \init)_{a^{-1}} - (U_{\sigma_a}^{t^*} \init)_{a^{-1}}| \\
&\leq |(\BM{\beta}_{+})_{a^{-1}} + (U_{\sigma_a}^{t^*} \init)_{a^{-1}}| + |(U_{\sigma_a}^{t^*} \init)_{a^{-1}}| \\
&= | (\BM{\beta}_{+} + U_{\sigma_a}^{t^*} \init)_{a^{-1}} | + \sqrt{q^*} \\
&\leq \| \BM{\beta}_{+} + U_{\sigma_a}^{t^*} \init \| + \sqrt{q^*} \\
&= \| U_{\sigma_a}^{t^*} \init - U_{\sigma_a}^{t^*} (i \BM{\beta}_-) + U_{\sigma_a}^{t^*} (i \BM{\beta}_-) + \BM{\beta}_{+}\| + \sqrt{q^*} \\
&\leq \| U_{\sigma_a}^{t^*} \init - U_{\sigma_a}^{t^*} (i \BM{\beta}_-) \| + \| U_{\sigma_a}^{t^*} (i \BM{\beta}_-) + \BM{\beta}_{+}\| + \sqrt{q^*} \\
&= \| \init - i \BM{\beta}_- \| + \| U_{\sigma_a}^{t^*} (i \BM{\beta}_-) + \BM{\beta}_{+}\| + \sqrt{q^*}.
\end{align*}
Thus,
\[ \sqrt{q^*} \geq |(\BM{\beta}_{+})_a| - \| U_{\sigma_a}^{t^*} (i \BM{\beta}_-) + \BM{\beta}_{+}\| - \| \init - i \BM{\beta}_- \|. \]
The right-hand side is exactly the same as that of Inequality~\eqref{0930-2} in the proof of Theorem~\ref{0930-1}.
Therefore, by following exactly the same argument as in that proof, we obtain the statement.
\end{proof}

We note from Lemma~\ref{0930-3} that the underlying edge $\{a, a^{-1}\}$ corresponding to the marked arc $a$ can also be found with high probability.
Indeed, the probability of finding this edge is $p^* + q^*$.
Hence, by Theorem~\ref{0930-1} and Lemma~\ref{0930-3}, we obtain the following.

\begin{cor}
For every integer $n\geq6$, the probability $p^*+q^*$ of finding the marked edge satisfies
\[ p^*+q^*\geq2P_n=1-\frac{6+4\sqrt2}{n}+\frac{17+12\sqrt2}{n^2}. \]
In particular, $p^*+q^*\geq1-\frac{6+4\sqrt2}{n}$ for every $n\geq2$.
\end{cor}

\begin{thm} \label{v2-sharp}
We have
\[ \lim_{n \to \infty} p^* = \lim_{n \to \infty} q^* = \frac{1}{2}. \]
More precisely,
\[ p^*=\frac12+O(n^{-1}),\qquad q^*=\frac12+O(n^{-1}),\qquad
1-(p^*+q^*)=O(n^{-1}). \]
\end{thm}

\begin{proof}
Take any constant $C > 3+2\sqrt2$.
By Theorem~\ref{0930-1}, Lemma~\ref{0930-3}, and $p^*+q^*\leq 1$,
\[ \frac{1}{2} - \frac{C}{n} \leq p^* \leq 1 - q^* \leq \frac{1}{2} + \frac{C}{n} \]
for sufficiently large $n$.
Therefore, we have
\[ \lim_{n \to \infty} p^* = \frac{1}{2}. \]
The assertion for $q^*$ can be proved in the same way.
The estimate $1-(p^*+q^*)=O(n^{-1})$ follows from the preceding corollary.
\end{proof}

From the above discussion, the appearance of $1/2$ in the probability of finding the marked arc $a$ in the complete bipartite graph $K_{n,n}$ can be explained by the fact that, at the optimal time, the probability is asymptotically concentrated equally on the marked arc $a$ and its inverse arc $a^{-1}$.
Accordingly, if we do not distinguish the direction and regard the underlying edge $\{a, a^{-1}\}$ as the search target, then its finding probability converges to $1$.

Figure~\ref{1019-1} presents the results of numerical computations for $2\leq n\leq200$.
This figure shows the actual success probability $p^*$, the lower bound obtained from the right-hand side of Lemma~\ref{0905-1}, and the polynomial lower bound $P_n$ in~\eqref{v2-polynomial-bound} for $n\geq6$; for $2\leq n\leq5$, the latter curve is set to the trivial lower bound $0$.
From the numerical computations, we observe that $p^*$ approaches $1/2$ as $n$ increases.
Moreover, the lower bound obtained from the right-hand side of Lemma~\ref{0905-1} provides a good approximation to $p^*$ from below.
The polynomial lower bound $P_n$ converges to $1/2$ at a rate of $1/n$, with all coefficients explicitly given.

\begin{figure}[htbp]
\centering
\includegraphics[width=0.8\linewidth]{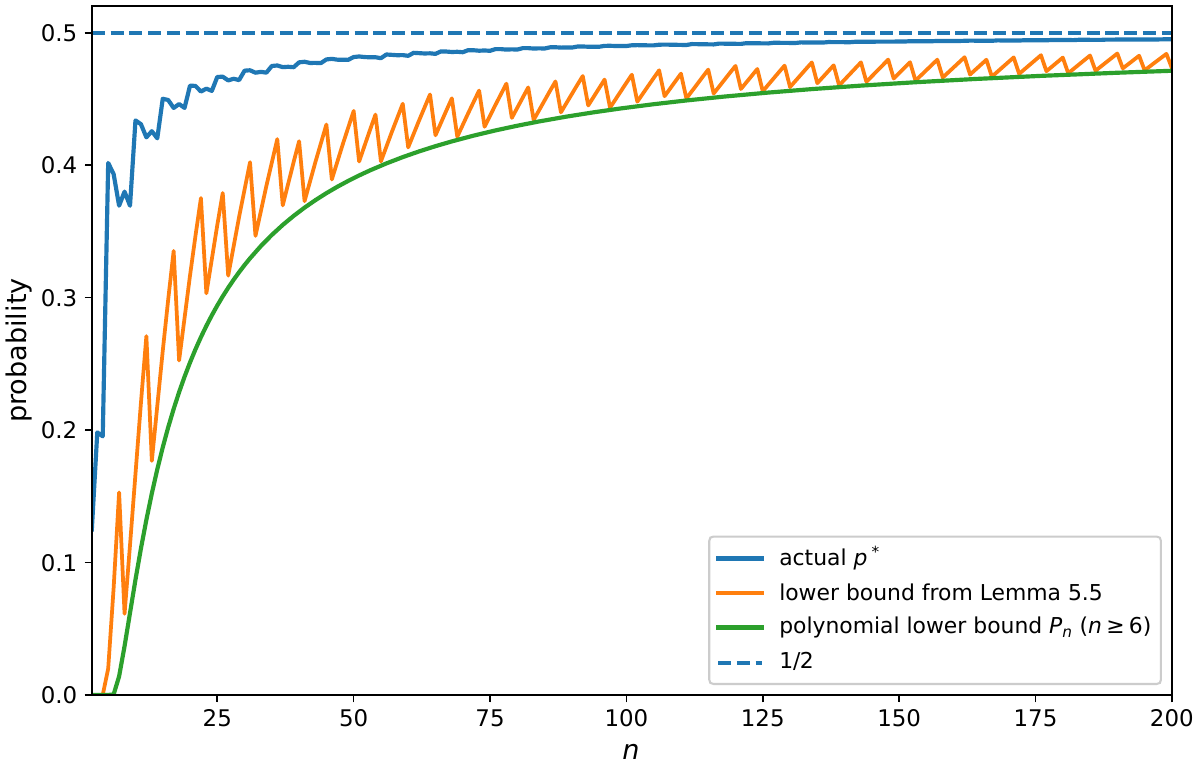}
\caption{Numerical computations for $2\leq n\leq200$.
Comparison of the actual success probability $p^*$, the lower bound obtained from Lemma~\ref{0905-1}, and the polynomial lower bound $P_n$ in~\eqref{v2-polynomial-bound} for $n\geq6$ (set to $0$ for $2\leq n\leq5$).
The dashed line represents the limiting value $1/2$.}
\label{1019-1}
\end{figure}

Note that Silva et al. also showed that the success probability approaches $1$ for edge search on complete bipartite graphs using the continuous-time total-graph model~\cite{silva2023walking}.
Although similar results for edge search are obtained in both models, a direct comparison is difficult because their state spaces and time evolutions are different.
However, it would be an interesting direction for future work to clarify whether there exist phenomena specific to each model, including the search model considered in this study, search models based on continuous-time quantum walks, and other search models.

\section{{Final remarks}}

One of the key factors that made this research possible, particularly the computations presented in Section~\ref{0918-1}, was the work of Akbari et al.~\cite{akbari2018spectrum} on the eigenvalue analysis of edge-signed graphs.
Their study is also utilized in the work of Segawa and Yoshie~\cite{segawa2021quantum}.
In future studies on quantum search for edges or arcs in various graphs,
it will be essential to begin with the eigenvalue analysis of edge-signed graphs.
Therefore, we believe that it is important to investigate in advance the eigenvalues and eigenvectors of various edge-signed graphs, also from the perspective of spectral graph theory.

\section*{Acknowledgements}
{We would like to thank the anonymous reviewers for their valuable and constructive comments, which helped improve the paper.}
This work is supported by JSPS KAKENHI (Grant Number JP24K16970).

\bibliographystyle{plain}
\bibliography{mybibs}

\end{document}